\documentclass[12pt,a4paper]{scrartcl}

\usepackage{makeidx}
\usepackage{comment}
\usepackage{algorithm}
\usepackage{algorithmicx}
\usepackage[noend]{algpseudocode}
\usepackage{url}
\usepackage{graphicx}
\usepackage{multirow}
\usepackage{amsmath}
\usepackage{latexsym}
\usepackage{subcaption}

\newtheorem{theorem}{Theorem}
\newtheorem{proof}{Proof}

\newtheorem{definition}{Definition}

\newtheorem{prob}{Problem}

\begin{document}

\title{\Large Online Grammar Compression for Frequent Pattern Discovery}
\author{\normalsize Shouhei Fukunaga, Yoshimasa Takabatake, Tomohiro I and Hiroshi Sakamoto\thanks{mailto: hiroshi@ai.kyutech.ac.jp}\\
{\small Kyushu Institute of Technology, Japan}
}
\date{\empty}
\maketitle

    \abstract{ 
Various grammar compression algorithms have been proposed in the last decade.
A grammar compression is a restricted CFG deriving the string deterministically.
An efficient grammar compression develops a smaller CFG by finding duplicated patterns and removing them.
This process is just a frequent pattern discovery by grammatical inference.
While we can get any frequent pattern in linear time using a preprocessed string,
a huge working space is required for longer patterns,
and the whole string must be loaded into the memory preliminarily.
We propose an online algorithm approximating this problem within a compressed space.
The main contribution is an improvement of the previously best known 
approximation ratio $\Omega(\frac{1}{\lg^2m})$ to $\Omega(\frac{1}{\lg^*N\lg m})$ 
where $m$ is the length of an optimal pattern in a string of length $N$ 
and $\lg^*$ is the iteration of the logarithm base $2$.
For a sufficiently large $N$, $\lg^*N$ is practically constant.
The experimental results show that our algorithm extracts nearly optimal patterns
and achieves a significant improvement in memory consumption compared to 
the offline algorithm.
}

\section{Introduction}
A \emph{grammar compression} of a string is a context-free grammar (CFG) that derives only the string.
In recent decades, various grammar compression algorithms have been proposed,
showing good performance, especially for a \emph{repetitive string}
in which long identical patterns (substrings) can be observed many times.
Such data are currently ubiquitous, for example, in genome sequences collected from similar species
and in versioned documents maintained by Wikipedia and GitHub, etc.
Because repetitive strings are growing rapidly,
data processing methods on grammar compression have been extensively studied 
as a promising way to address repetitive strings
(\emph{e.g.}, \cite{Larsson00,Lehman-Shelat02,Yang03,Rytter03,Sakamoto05,Charikar05,Maruyama2012,Maruyama2013,Tabei13,Maruyama2014}).

\emph{Frequent pattern discovery} is a classic problem in pattern mining for sequence data (\emph{e.g.}, \cite{FPM}),
where we focus on a string and say that a pattern (substring) is frequent if it occurs at least twice.
Longer patterns are often the target of discovery, as they seem to characterize the input string better.
Although we have linear time solutions using a full-text index such as suffix tree and suffix array by \cite{Sadakane00},
it requires a huge working space for large-scale data.
Even if we opt for space-efficient alternatives of these data structures, such as FM-index by \cite{Ferragina00},
we still have to load the whole string into memory, at least at the construction phase,
because there are no known algorithms to construct them in a streaming fashion.
Due to these drawbacks, it is difficult to apply these algorithms to stream data.

A reasonable approach to avoid this difficulty is to seek an approximate frequent pattern instead of the exact solution.
In the framework of grammar compression, an approximate pattern is found as a frequent subtree.
Then, a suitable parsing tree should preserve as many occurrences of a common substring as possible.
Edit-sensitive parsing (ESP) by \cite{Cormode07} matches the claim;
ESP approximately solves the NP-hard problem of 
the generalized edit distance for measuring the similarity of two strings, and
online algorithms and applications of ESP were widely proposed (\emph{e.g.}, \cite{Hach2012,ESP,Takabatake14-2,Takabatake15,Takabatake16,Nishimoto2016}).

As seen above, grammar compression is closely related to the approximate pattern discovery
because a good compression ratio is achieved by finding frequent substrings and replacing them 
by a variable that derives the substrings.
\cite{Nakahara13} focused on a grammar compression algorithm (called ESP-comp) based on ESP
and showed that it approximately solves the frequent pattern discovery problem.
That is, they showed that for any frequent pattern $P$,
there is a variable $X$ such that
(1) $X$ derives a string of length $\Omega(\frac{|P|}{\lg^2|P|})$ that is a substring of $P$ and
(2) $X$ accompanies any occurrence of $P$ in the string.
They confirmed by computational experiments that the algorithm
efficiently finds long frequent patterns from large repetitive data.

In this paper, we follow the previous work and show a new lower bound $\Omega(\frac{1}{\lg^*N\lg |P|})$ for approximation, 
where $N$ is the length of the string and $\lg^*$ is the iteration of logarithm base $2$.
This improves the previous bound $\Omega(\frac{1}{\lg^2|P|})$ as $\lg^*N \leq \lg |P|$ in practice.
In addition, we establish an online approximation algorithm within a compressed space using ESP-comp.
Note that the previous algorithm was not online, \emph{i.e.}, the whole string must be loaded into memory,
but recent progress by \cite{Maruyama2013} has enabled the computation of ESP-comp in compressed space in a streaming fashion.
We implement our algorithm and show experimentally that the approximation is nearly optimal 
and the improvement of memory consumption is significant for real data.

\section{Definition}
\subsection{Notation}

Let $\Sigma$ be a finite alphabet, and $\sigma$ be $|\Sigma|$.  All elements
in $\Sigma$ are totally ordered.  Let us denote by $\Sigma^*$ the set of all
strings over $\Sigma$, and by $\Sigma^q$ the set of strings of length $q$ over
$\Sigma$, \emph{i.e.}, $\Sigma^q=\{w\in \Sigma^*:|w|=q\}$ and an element in
$\Sigma^q$ is called a $q$-gram.  The length of a string $S$ is denoted by
$|S|$.  The empty string $\epsilon$ is a string of length $0$, namely
$|\epsilon|=0$.  For a string $S=\alpha\beta\gamma$, $\alpha$, $\beta$ and
$\gamma$ are called the prefix, substring, and suffix of $S$, respectively.
The $i$-th character of a string $S$ is denoted by $S[i]$ for $i \in [1,|S|]$.
For a string $S$ and interval $[i,j]$ ($1 \leq i \leq j \leq |S|$), let
$S[i,j]$ denote the substring of $S$ that begins at position $i$ and ends at
position $j$, and let $S[i,j]$ be $\epsilon$ when $i>j$.
For a string $S$ and integer $q \geq 0$, let
$\mathit{pre}(S,q)=S[1,q]$ and $\mathit{suf}(S,q)=S[|S|-q+1,|S|]$.
For strings $S$ and $P$, let $\mathit{freq_S(P)}$ denote the number of occurrences of $P$ in $S$,
\emph{i.e.}, $\mathit{freq_S(P)} = |\{ i : S[i, i+|P|-1] = P \}|$.
We assume a recursive enumerable set
${\cal X}$ of variables with $\Sigma\cap {\cal X}=\emptyset$.  All elements in
$\Sigma\cup {\cal X}$ are totally ordered, where all elements in $\Sigma$ must
be smaller than those in ${\cal X}$.  In this paper, we call a sequence of
symbols from $\Sigma \cup X$ a string.  Let us define $\lg^{(1)}{u}=\lg{u}$,
and $\lg^{(i+1)}{u}=\lg{(\lg^{(i)}{u})}$ for $i\geq 1$.  The iterated
logarithm of $u$ is denoted by $\lg^*{u}$, and defined as the number of times
the logarithm function must be applied before the result is less than or equal
to $1$, \emph{i.e.}, $\lg^*u=\min\{i:\lg^{(i)}{u} \leq 1\}$.

\subsection{Grammar Compression}

We consider a special type of context-free grammar (CFG) $G=(\Sigma,V,D,X_s)$
where $V$ is a finite subset of ${\cal X}$, 
$D$ is a finite subset of $V\times (V\cup \Sigma)^*$, and $X_s \in V$ is the start symbol.
A grammar compression of a string $S$ is a CFG deriving only $S$ deterministically, \emph{i.e.},
for any $X\in V$ there exists exactly one production rule in $D$ and there is no loop.
Because each $G$ has its Chomsky normal form,
we can assume that any grammar compression is in \emph{Straight-line program (SLP)} by \cite{SLP}:
any production rule is in the form of $X_k\to X_iX_j$ where 
$X_i,X_j \in \Sigma\cup V$ and $1\leq i,j<k\leq n + \sigma$. 

The size of an SLP is the number of variables, \emph{i.e.}, $|V|$ and let $n=|V|$.
$\mathit{val}(X_i)$ for variable $X_i \in V$ denotes the string derived from $X_i$.
For $w \in (V\cup \Sigma)^*$, let $\mathit{val}(w) = \mathit{val}(w[1]) \cdots \mathit{val}(w[|w|])$.

The parse tree of $G$ is a rooted ordered binary tree such that 
(i) the internal nodes are labeled by variables and (ii) the leaves are labeled by alphabet symbols.
In a parse tree, any internal node $Z$ corresponds to a production rule $Z\to XY$,
and has the left child with label $X$ and the right child with label $Y$.

A phrase dictionary $D$ is a data structure for directly accessing the phrase $X_iX_j$ for any $X_k$ 
if $X_k\to X_iX_j$ exists.
On the other hand, a reverse dictionary $D^{-1}$ is a data structure for directly accessing $X_k$
for $X_iX_j$ if $X_k\to X_iX_j$ exists.

\subsection{Succinct Data Structure}
A grammar compression is encoded by succinct data structures.
A rank/select dictionary for a bit string $B$ by \cite{Jacobson89} supports the
following queries: $\mathit{rank}_c(B,i)$ returns the number of occurrences of
$c \in \{0,1\}$ in $B[0,i]$; $\mathit{select}_c(B,i)$ returns the position of
the $i$-th occurrence of $c\in\{0,1\}$ in $B$; $\mathit{access}(B,i)$ returns
the $i$-th bit in $B$.  Data structures with only the $|B| + o(|B|)$ bits
storage to achieve $O(1)$ time rank and select queries have been presented by \cite{Raman07}.
Wavelet tree is an extension of the dictionary over string proposed by \cite{Grossi03}.
\cite{Golynski06} proposed an improvement of the wavelet tree, called GMR, 
in $(n+\sigma)\lg{(n+\sigma)}+o((n+\sigma)\lg{(n+\sigma)})$ bits while
supporting both rank and access queries in $O(\lg{\lg{(n+\sigma)}})$ times and select queries in $O(1)$ time.

\subsection{Approximate Frequent Pattern}
A substring $P=S[i,j]$ is said to be frequent if it appears at least twice, \emph{i.e.}, $\mathit{freq_S(P)}\geq 2$.
We focus on an approximation of the problem to find all frequent patterns defined as follows.

\begin{prob} \label{prob:1}
Let $T$ be a parsing tree of a grammar compression deriving $S\in\Sigma^*$.
A variable $X$ in $T$ is called a core of $P$ if 
for each occurrence $S[i,j]=P$, there exists an occurrence of $X$ in $T$ deriving a substring $S[\ell,r]$
for a subinterval $[\ell,r]$ of $[i,j]$.
Then, $P$ is said to be approximated by $X$ with $\delta$ if $\frac{|val(X)|}{|P|}\geq\delta$.
The problem of approximated frequent pattern (AFP) is to compute $T$ 
that guarantees a core $X$ of any frequent pattern $P$ in $S$ with an approximation ratio $\delta>0$.
\end{prob}

AFP is well-defined with a small $\delta$ because for any $S$ and its frequent substring $P$
any alphabet symbol forming $P$ satisfies the condition with $\delta=\frac{1}{|P|}$.
\cite{Nakahara13} proposed an offline algorithm with approximation $\Omega(\frac{1}{\lg^2|P|})$.
We aim to construct the parsing tree by an online algorithm in a compressed space with a larger $\delta$
improving the best known approximation ratio.
In our algorithm, a grammar compression is represented by ESP (edit sensitive parsing) 
and succinctly encoded by POSLP (post-order SLP).
We next review the related techniques.

\subsection{Edit Sensitive Parsing}
Originally, ESP (Edit Sensitive Parsing) was introduced 
by \cite{Cormode07} and widely applied in data compression and information retrieval
(e.g., \cite{Hach2012,Takabatake14-2,Takabatake15,Takabatake16,Nishimoto2016}).
ESP is a parsing technique intended to efficiently construct a consistent parsing for same substrings as follows.

For each substring $S[i,j]$, we can decompose it into a sequence of
subtrees rooted by symbols $X_1,X_2,\ldots,X_q$.
For each frequent $P$ (e.g., $S[i,j] = S[k,\ell] = P)$,
we can find a consistent decomposition for the occurrences
by the trivial decomposition of $X_1=S[i],X_2=S[i+1],\ldots,X_p=S[j]$.
For this problem, ESP tree guarantees a better decomposition:
$(X_1,X_2,\ldots,X_q)$ is embedded into any occurrence of $P$ with a small $q$
e.g., \cite{Nakahara13} showed that $q=\Omega(|P|/\lg^2|P|)$
and we improve it to $\Omega(|P|/\lg^*|P|\lg|P|)$ in this paper.
Any symbol in the decomposition is regarded to a necessary condition of an occurrence of $P$.
Using this fact, we can find an approximate pattern from ESP tree.
Using this result, we can efficiently compute a smaller grammar compression closely related to AFP.

We review the algorithm for ESP presented in \cite{Takabatake16}.
This algorithm, referred to as ESP-comp, computes an SLP from an input string $S$.
The tasks of ESP-comp are to (i) partition $S$ into $s_1s_2\cdots s_\ell$ such that
$2\leq |s_i|\leq 3$ for each $1\leq i\leq \ell$,
(ii) if $|s_i|=2$, generate the production rule $X\to s_i$ and replace $s_i$ by $X$
(this subtree is referred to as a 2-tree), and 
if $|s_i|=3$, generate the production rule $Y\to AX$ and $X\to BC$ for $s_i=ABC$,
and replace $s_i$ by $Y$ (referred to as a 2-2-tree),
(iii) iterate this process until $S$ becomes a symbol.
Finally, the ESP-comp builds an SLP representing the string $S$.

We focus on how to determine the partition $S=s_1s_2\cdots s_\ell$.
A string of the form $a^r$ with $a\in\Sigma\cup V$ and $r\geq 2$ is called a repetition.
A repetition $S[i,j]$ is called to be \emph{maximal} if $S[i] \neq S[i-1],S[j+1]$.
First, $S$ is uniquely partitioned into the form
$w_1x_1w_2x_2\cdots w_kx_kw_{k+1}$  by its maximal repetitions, where each $x_i$ is 
a maximal repetition of a symbol in $\Sigma\cup V$, and each $w_i\in(\Sigma\cup V)^*$ contains no repetition.
Then, each $x_i$ is called type1, each $w_i$ of length at least $2\lg^*|S|$ is type2, and any remaining $w_i$ is type3.
If $|w_i|=1$, this symbol is attached to $x_{i-1}$ or $x_i$ with preference $x_{i-1}$ when both cases are possible.
Thus, if $|S|>2$, each $x_i$ and $w_i$ is longer than or equal to two.

Next, ESP-comp parses each substring $v$ depending on the type.
For type1 and type3 substrings, the algorithm performs the \emph{left aligned parsing} as follows. 
If $|v|$ is even, the algorithm builds 2-tree from $v[2j-1,2j]$ for each $j \in \{1,2,\ldots,|v|/2\}$;
otherwise, the algorithm builds a 2-tree from $v[2j-1,2j]$ for each $j \in \{1,2,\ldots,\lfloor(|v|-3)/2\rfloor\}$ and 
builds a 2-2-tree from the last trigram $v[|v|-2,|v|]$. 
If $v$ is type2, the algorithm further partitions it into short substrings of length two or three by the following
\mbox{\emph{alphabet reduction}}.

{\bf Alphabet reduction:} 
Given a type2 string $v$, consider $v[i]$ and $v[i-1]$ as binary integers.
Let $p$ be the position of the least significant bit of $v[i]\oplus v[i-1]$
and let $bit(p,v[i])$ be the bit of $v[i]$ at the $p$-th position.
Then, $L(v)[i]=2p+bit(p,v[i])$ is defined for any $i\geq 2$.
Because $v$ is repetition-free (\emph{i.e.}, type2), the label string $L(v)[2,|v|]$ is also type2.
Suppose that any symbol in $v$ is an integer in $\{0, \ldots, N\}$,
$L(v)[2,|v|]$ is a sequence of integers in $\{0, \ldots, 2 \lg N + 1\}^{|v|-1}$.
If we apply this procedure $\lg^*N$ times, 
then we get $L^*(v)[\lg^*N + 1, |v|]$ a sequence of integers in $\{0, \ldots, 5\}^{|v|-\lg^*N}$,
where $L^*(v)[1, \lg^*N]$ is not defined
\footnote{The number of iteration of alphabet reduction should not be changed arbitrarily according to each $v$, and so
$N$ is set in advance to be a sufficiently large integer, e.g. $N = O(|S|)$.}.
When $L^*(v)[i-1],L^*(v)[i],L^*(v)[i+1]$ are defined, 
$v[i]$ is called the \emph{landmark} if $L^*(v)[i]> \max\{L^*(v)[i-1],L^*(v)[i+1]\}$.

The iteration of alphabet reduction transforms $v$ into $L^*(v)$ such that any substring of $L^*(v)[\lg^*N + 1, |v|]$ of 
length at least $12$ contains at least one landmark because $L^*(v)[\lg^*N + 1, |v|]$ is also type2.
Using this characteristic, the algorithm ESP-comp determines the bigrams $v[i,i+1]$ to be replaced
for any landmark $v[i]$, where any two landmarks are not adjacent, and then the replacement is deterministic.
After replacing all landmarks, any remaining maximal substring $s$ is replaced by the left aligned parsing,
where if $|s|$ =1, it is attached to its left or right block.

We give an example of the edit sensitive parsing of an input string in Figure~\ref{fig:reduction}-(i) and (ii).
For type2 substring $v$ (Figure~\ref{fig:reduction}-(i)),
$v$ is parsed according to landmarks (here landmarks are determined by conducting alphabet reduction two times for simplicity of explanation).
Any other remaining substrings including type1 and type3 are parsed by the left aligned parsing (shown in Figure~\ref{fig:reduction}-(ii)).
In this picture, a dashed node denotes that it is an intermediate node in a 2-2-tree.
Originally, an ESP tree is a ternary tree in which each node has at most three children.
The intermediate node is introduced to represent ESP tree as a binary tree.

The following characteristics are well-known for ESP.
By Theorem~\ref{th:cormode1}, we can obtain the locally consistent parsing for $S$:
An iteration of ESP for $S$, for any substring $P$ of $S$
there exists an interval $[i,j]$ of length at least $|P|-O(\lg^*|S|)$ 
such that the substring $P[i,j]$ with each occurrence of $P$ 
is transformed into a same string.
Iterating this, the resulting ESP tree contains a large subtree for $P$
regardless of its occurrence, that expresses an approximation of $P$.
Theorem~\ref{th:cormode2} is clear by the definition of ESP.
Adopting such theorems, we derive our results in the following section.

\begin{theorem}\label{th:cormode1}(\cite{Cormode07})
For type2 substring $v$, whether $v[i]$ is a landmark or not is determined by only $v[i-O(\lg^*|S|),i+O(1)]$.
\end{theorem}
\begin{theorem}\label{th:cormode2}(\cite{Cormode07})
The height of ESP tree of $S$ is $O(\lg|S|)$.
\end{theorem}

\begin{figure}[t]
\begin{center}
\includegraphics[width=0.7\textwidth]{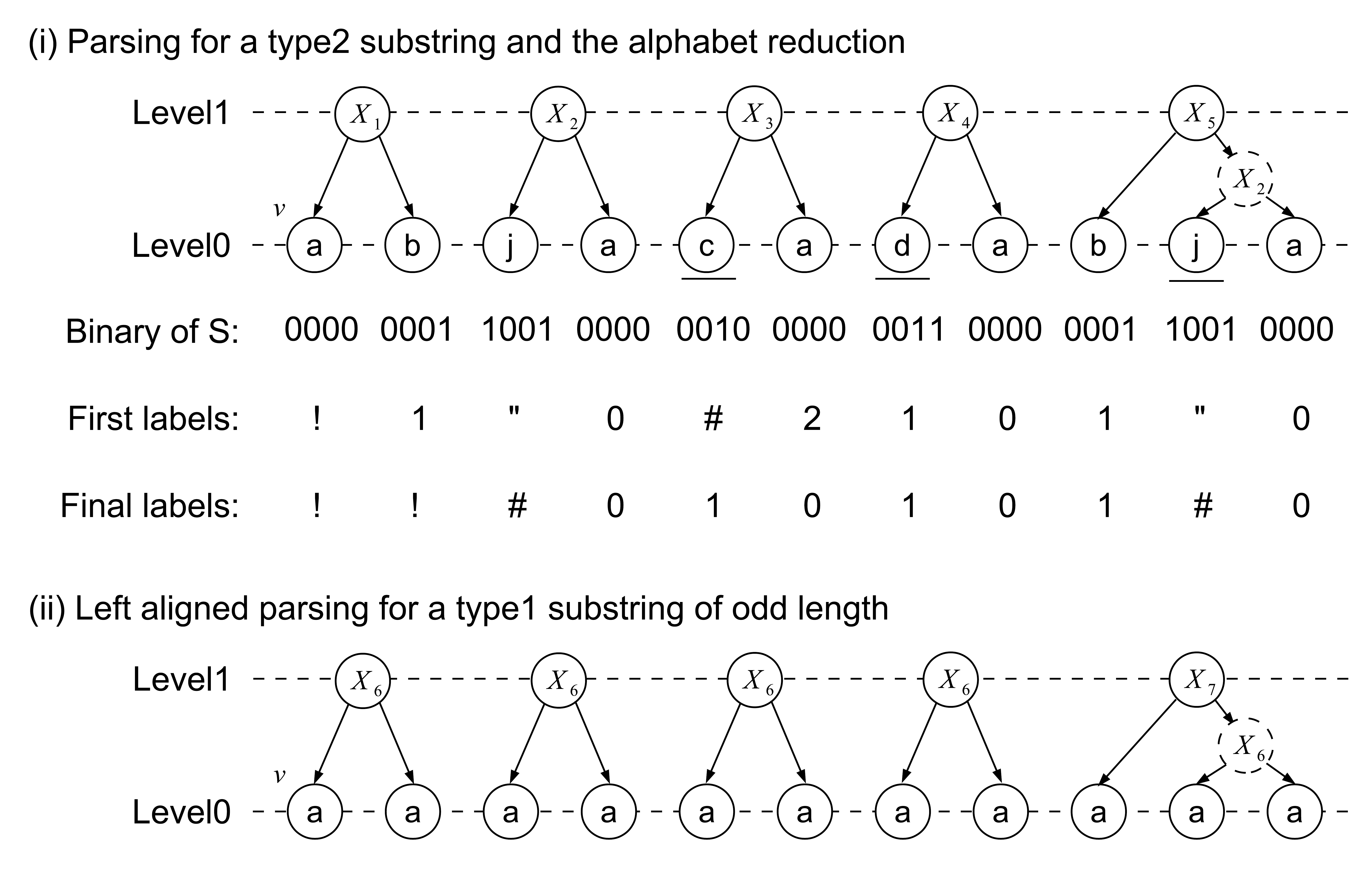}
\end{center}
\vspace{-0.6cm}
\caption{The edit sensitive parsing. In (i), an underlined $v[i]$ means a landmark, and $p\geq 0$.
In (i) and (ii), a dashed node is corresponding to the intermediate node in a 2-2-tree.}\vspace{-12pt}
\label{fig:reduction}
\end{figure}

\subsection{Succinct Encoding}
\begin{figure*}[t]
\begin{center}
\includegraphics[width=0.98\textwidth]{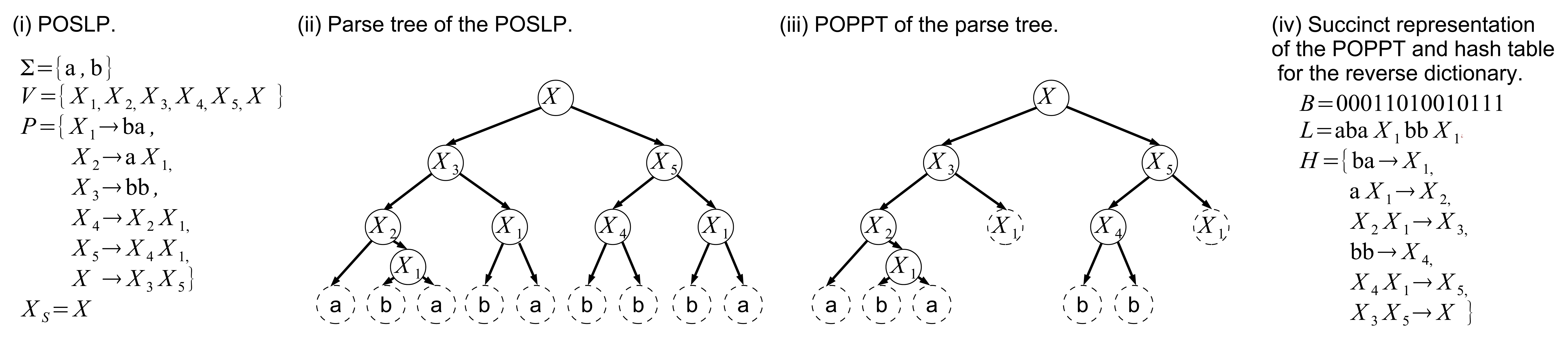}
\end{center}
\vspace{-0.6cm}
\caption{Example of post order SLP (POSLP), parse tree, post order partial parse tree (POPPT), and succinct representation of POPPT.}
\label{fig:pslp}
\end{figure*}

\cite{Rytter03} defined a partial parse tree as a binary tree built by traversing a parse tree in a depth-first manner 
and pruning out all the descendants under every node of a nonterminal symbol appearing before.
\cite{Maruyama2012} introduced the post-order SLP (POSLP) and post-order partial parse tree (POPPT) as follows.

\begin{definition}[POSLP and POPPT]\label{def:pslp}
A post-order partial parse tree is a partial parse tree whose internal nodes have post-order variables.
A post-order SLP is an SLP whose partial parse tree is a post-order partial parse tree.
\end{definition}

For a POSLP of $n$ variables, the number of nodes in the POPPT is $2n+1$ because 
the numbers of internal nodes and leaves are $n$ and $n+1$, respectively.
Figure~\ref{fig:pslp}-(i)(iii) shows an example of POSLP and POPPT, respectively.
The resulting POPPT (iii) has internal nodes consisting of post-order variables.

\cite{Maruyama2013} proposed FOLCA, the fully online algorithm for computing succinct POSLP $(B,L)$.
$B$ is the bit string obtained by traversing POPPT in post-order, and putting $'0'$ if a node is a leaf and $'1'$ otherwise.
The last bit $'1'$ in $B$ represents the super root.
$L$ is the sequence of leaves of the POPPT.
The dynamic sequences $B$ and $L$ are encoded using the succinct data structure by \cite{NavSad12}.
Then, the following result was shown.

\begin{theorem}[\cite{Maruyama2013}]\cite{}\label{FOLCA}
\label{th:maruyama}
The POSLP of $n$ variables and $\sigma$ alphabet symbols supporting
the phrase and reverse dictionaries can be constructed 
in $O(\frac{|S|\lg n}{\alpha \lg\lg n})$ expected time using 
$(1+\alpha)n\lg (n+\sigma) + n(3+\lg(\alpha n))$ bits memory where
$\alpha\in (0,1)$ is the load factor of a hash table.
\end{theorem}

In this paper, we improve the approximation ratio of the pattern discovery problem
and show the ratio is practically sufficient (nearly optimal) for real data.
The implementation of the algorithm is realized by the modification of FOLCA 
with novel data structures.

\section{Algorithm}

\begin{algorithm}[htbp]
  \caption{ to compute a core $X$ of any frequent $P$ in $S$. 
	$T$: POSLP representing the ESP tree of $S$,
	$B$: a succinct representation of skeleton of $T$,
	$L$: a sequence of leaves of $T$,
	$FB$: a bit vector storing $FB[i]=1$ iff $\mathit{freq}_T(X_i)\geq 2$,	$D^{-1}$: the reverse dictionary for production rules,
	$q_k$: a queue in $k$-th level,
	and let $u\in \max\{5, \lg^*|S|\}$.} 
  \label{algo}
{\footnotesize 
\begin{algorithmic}[1]
\Function{ComputeAFP}{$S$}
\State $B:=\emptyset; L:= \emptyset; FB:=\emptyset;$ initialize queues $q_k$
\For{$i:=1,2,\ldots,|S|$}
\State {\sc BuildESPTree}($\{S[i],0,0,0,0\},q_1$)
\EndFor
\EndFunction
\Function{BuildESPTree}{$X,q_k$} \Comment{$X$ is a set $\{s,ib,\ell_1,\ell_2,\ell_{\lg^*|S|}\}$ where $s$ is a symbol, $ib$ is 1 if $s$ is an internal node otherwise 0 and $\ell_i (i \in \{1,2,\lg^*|S|\})$ is a label applied $i$-th alphabet reduction for $s$.}
\State $q_k.enqueue(X)$
\State compute $q_k[q_k.length()].\ell_i(i \in \{1,2,\lg^*|S|\})$ 
\If{$q_k.length() = u$}
\If{{\sc Is2Tree}($q_k$)}
\State $Y:=$ {\sc Update}($q_k[u-1],q_k[u]$)  
\State $q_k.dequeue(); q_k.dequeue()$
\State {\sc BuildESPTree}($Y,q_{k+1}$)
\EndIf
\ElsIf{$q_k.length() = u+1$}
\State $Y:=$ {\sc Update}($q_k[u],q_k[u+1]$); $Z:=$ {\sc Update}($q_k[u-1],Y$)
\State $q_k.dequeue(); q_k.dequeue(); q_k.dequeue()$
\State {\sc BuildESPTree}($Z,q_{k+1}$)
\EndIf
\EndFunction
\Function{Is2Tree}{$q_k$}
\If{($q_k[u-4].s= q_k[u-3].s) \& (q_k[u-3].s \neq q_k[u-2].s)$}
\State \Return 0
\ElsIf{($q_k[u-3].s\neq q_k[u-2].s) \& (q_k[u-2].s = q_k[u-1].s)$}
\State \Return 0
\ElsIf{$(q_k[u-3].\ell_{\lg^*|S|} < q_k[u-2].\ell_{\lg^*|S|}) \& (q_k[u-2].\ell_{\lg^*|S|} > q_k[u-1].\ell_{\lg^*|S|})$}
\State \Return 0
\Else
\State \Return 1
\EndIf
\EndFunction
\Function{Update}{$X,Y$}
\State $z:= D^{-1}(X.s,Y.s)$ 
\If {$z$ is a new symbol}
\State {\sc UpdateLeaf}($X$); {\sc UpdateLeaf}($Y$)
\State $B.push\_back(1); FB.push\_back(0)$
\State \Return $\{z,1,0,0,0\}$ 
\Else
\State {\sc GetAFPNode}($z$)
\State \Return $\{z,0,0,0,0\}$ 
\EndIf
\EndFunction
\Function{UpdateLeaf}{$X$}
\If {$X.ib=0$}
\State $L.push\_back(X.s); B.push\_back(0)$
\EndIf
\EndFunction
\Function{GetAFPNode}{$X_i$}
\If{$FB[i] = 0$}
\State $FB[i] := 1$
\State Output $X_i$
\EndIf
\EndFunction
\end{algorithmic}
}
\end{algorithm}

In this section, we propose a modified FOLCA for AFP with saving-space.
We show the improved lower bound of the size of extracted core
as well as time and space complexities.
We first summarize the proposed algorithm.
Let $S_i$ $(i=0,1,\ldots,\lceil \lg |S| \rceil)$ be the resulting string of the $i$-th iteration of ESP, where $S_0=S$.
The algorithm simulates the parsing of ESP using a queue $q_i$ for each level $i$.
The queue $q_i$ stores a substring $S_i$ of length at most $O(\lg^*|S|)$ in a FIFO manner.
At the beginning, input symbols are enqueued to $q_0$.
If a prefix of $S$ is a repetition $a^+$, it is parsed in a left-aligned manner, and a production rule such as $A\to aa$ is generated.
$a^+$ is dequeued from $q_0$, and the resulting sequence of $A$s is enqueued to $q_1$.
Otherwise, at most $O(\lg^*|S|)$ symbols are enqueued to $q_0$, and 
$q_0[0,i-1]$ is parsed in a left-aligned manner, where $q_0[i]$ is the leftmost landmark.
By Theorem~\ref{th:cormode1}, there is at least one landmark in $q_0$ of length $O(\lg^*|S|)$.
Then, the symbols in $q_0[0,i-1]$ are dequeued from $q_0$, and the generated symbols are enqueued to $q_1$.
These computations are done in each level.
When a prefix of $S$ is enqueued, a sequence of production rules is generated such that
it is encoded by a POSLP $T$ encoded by $(B,L)$, where $B$ is a bit sequence that represents the skeleton of $T$,
and $L$ is the sequence of the leaves of $T$.
The pseudo code is shown in Algorithm~\ref{algo}.

We next show that the ESP tree of $S$ contains a sufficiently large core for any substring $P$
that guarantees the approximation ratio of our algorithm.
This result is an improvement of the lower bound shown by \cite{Nakahara13}.

\begin{theorem}\label{th:approx}
Let $T$ be the ESP tree of a string $S$ and $P$ be a substring of $S$.
There exists a core of $P$ that derives a string of length $\Omega(\frac{|P|}{\lg^*|S|\lg|P|})$.
\end{theorem}
\begin{proof}
If a prefix of $P$ is a repetition, let $Q_1$ be the maximal one and $Q'_1$ be the remaining suffix of $P$.
The parsing of $Q'_1$ is not affected by the string preceeding $Q'_1$, and then
the parsing of $Q'_1$ inside $P$ is identical regardless of any occurrence of $P$. 
Otherwise, by Theorem~\ref{th:cormode1}, we can partition $P=Q_1Q'_1$ such that
$|Q_1|=O(\lg^*|S|)$, and $Q'_1$ is also identically parsed inside $P$.
Let $P_1$ be the common substring in $S_1$ deriving $Q'_1$.
Then, for each case, $Q_1P_1$ is a sequence of cores of $P$.
Iterating this process for $P_1$ at most $k (\leq \lceil \lg |P| \rceil)$ times, we can get a sequence $Q_1Q_2\cdots Q_k$ of cores 
such that $Q_i$ is either a repetition of the form $Q_i=c_i^+$ $(c_i\in\Sigma\cup V)$ or a string of length $O(\lg^*|S|)$.

We show that for any $1 \leq i \leq k$, there exists a core $X_i$ in $Q_i$ with $|\mathit{val}(X_i)| = \Omega(\frac{|\mathit{val}(Q_i)|}{\lg^*|S|})$.
If the length of $Q_i$ is $O(\lg^*|S|)$, the claim is immediate from the pigeonhole principle.
Otherwise $Q_i=c_i^+$.
Because any maximal repetition is parsed in a left-aligned manner, 
a type2 sequence of bigrams $c_i^2$ is created over $Q_i$ (except for the last one, which may be a 2-2-tree deriving $c_i^3$).
Iterating the parsing on the type2 sequence, we will get a large complete balanced binary tree of $c_i$.
Assuming that the largest one covers $2^{h}$ $c_i$'s in $Q_i$,
we can see that $Q_i$ contains $c_i$'s less than $5 \cdot 2^{h}$, namely, there is a node covering at least one-fifth of the $c_i$'s in $Q_i$.
The maximum length of $Q_i$ is achieved when $Q_i$ is parsed into $AB C_{h-1} \cdots C_{0}$,
where $A$ contains $2^{h} - 1$ $c_i$'s, $B$ contains $2^{h}$ $c_i$'s, and for any $0 \leq h' < h$ $C_{h'}$ contains $3 \cdot 2^{h'}$ $c_i$'s.
$A$ and its preceeding character $c \neq c_i$ (that must be the first character in the whole string)
compose a node having $2^{h}$ characters, $B$ composes the largest complete binary tree with $2^{h}$ $c_i$'s,
and for any $0 \leq h' \leq h$ $C_{h'}$ composes a 2-2-tree over three complete binary trees with $2^{h'}$ $c_i$'s.
Note that adding another $c_i$ to the $Q_i$ results in creating a 2-2-tree over three complete binary trees with $2^{h}$ $c_i$'s
in which we have a complete binary tree with $2^{h+1}$ $c_i$'s, and thus, 
the maximum number of $c_i$'s in $Q_i$ is $2^{h} - 1 + 2^{h} + \sum_{h' = 0}^{h-1} 3 \cdot 2^{h'} < 5 \cdot 2^{h}$.
Therefore, there exists a variable $X_i$ in $Q_i$ with $|\mathit{val}(X_i)| = \Omega(\frac{|\mathit{val}(Q_i)|}{\lg^*|S|})$.

Because there is at least one $Q_j$ such that $|\mathit{val}(Q_j)| \geq |P| / k \geq |P| / \lg |P|$,
there exists a core of $P$ that derives a string of length $\Omega(\frac{|\mathit{val}(Q_j)|}{\lg^*|S|}) = \Omega(\frac{|P|}{\lg^*|S| \lg|P|})$.
\end{proof}

\begin{theorem}\label{th:freq}
Algorithm~\ref{algo} approximates the problem of AFP with the ratio $\Omega(\frac{1}{\lg^*|S|\lg|P|})$ 
in $O(\frac{|S|\lg n}{\alpha\lg\lg n})$ time and $O(n+\lg|S|)$ space.
\end{theorem}
\begin{proof}
The algorithm simulates the ESP of $S$ using queues $q_i$ $(i=0,1,\ldots,|S|)$;
$q_i$ stores a substring of $S_i$ to determine whether $S_i[j]$ is a landmark or not.
By Theorem~\ref{th:cormode1}, the space for each $q_i$ is $O(\lg^*|S|)$.
We can reduce this space to $O(1)$ using a table of size at most $\lg^*|S|\lg\lg\lg|S|$ bits as follows.
Applying two iterations of alphabet reduction, each symbol $A$ is transformed into a label $L_A$ of size at most $\lg\lg\lg|S|$ bits.
Whether the $A$ is a landmark or not depends on its consecutive $O(\lg^*|S|)$ neighbours.
Thus, the size of a table storing a 1-bit answer is at most $\lg^*|S|\lg\lg\lg|S|$ bits.
It follows that the space for parsing $S$ is $O(\lg|S|)$.
On the other hand, by Theorem~\ref{th:maruyama}, the POSLP $T$ of $S$ is computable in $O(\frac{|S|\lg n}{\alpha\lg\lg n})$ time.
By Theorem~\ref{th:approx}, for each frequent $P$, 
$T$ contains at least one core $X$ of $P$ satisfying $|\mathit{val}(X)|=\Omega(\frac{|P|}{\lg^*|S|\lg|P|})$.
Thus, finding all variables $X$ appearing at least twice in $T$ approximates this problem with the lower bound.
Whether $\mathit{freq}_T(X_i)\geq 2$ can be stored in $n$ bits for all $i$ because 
an internal node $i$ of $T$ denotes the position of the first occurrence of $X_i$.
Therefore, we obtain the complexities and approximation ratio.
\end{proof}

\section{Experimental Results}

We evaluate the performance of the proposed approximation algorithm 
on one core of a quad-core Intel Xeon Processor E5540 (2.53GHz) machine with 144GB memory. 
We adopt a lightweight version of the fully-online ESP, called FOLCA~\cite{Maruyama2013},
as a subroutine for the grammar compression.

We use several standard benchmarks from a text collection\footnote{
http://pizzachili.dcc.uchile.cl/repcorpus.html}, which is detailed in Table~\ref{tab:txtinfo}.
We choose texts with a high and small amount of repetitions.
For these texts, we examine the practical approximation ratio of the algorithm as follows.
For each text $S$, we obtained the set of frequent substrings by the compressed suffix array (SA) by \cite{Sadakane00},
and we selected the top-100 longest patterns so that any two $P$ and $Q$ are not {\em inclusive} of each other, 
where $P$ is inclusive of $Q$ if any occurrence of $Q$ is included in an occurrence of $P$.
We removed such $Q$ from the candidates.
For each frequent substring $P$ and a variable $X$ reported by the algorithm,
we estimate the cover ratio $\frac{|val(X)|}{|P|}$ and show the average for all $P$.
However, as shown in the result below  (Figure~\ref{fig:mem}), 
the suffix array cannot be executed for a larger $S$ due to memory consumption
Additionally, we examined the time and memory consumption of the offline algorithm by~\cite{Nakahara13}.

Table~\ref{tab:result} shows the length of optimum frequent patterns extracted by suffix array
and the length of the corresponding cores extracted by our algorithm as well as the approximation ratio to the optimal one,
where min./max. denote the shortest/longest pattern in the candidates, respectively.
Our algorithm extracted sufficiently long cores for each benchmark.

Figure~\ref{fig:mem} shows the memory consumption for repetitive strings 
(Figure~\ref{fig:mem_ein}-\ref{fig:mem_ker}) and normal strings (Figure~\ref{fig:mem_eng}-\ref{fig:mem_sou}).
The working space was significantly saved by our online strategy, where
offline and SA were executed for each static size of data noted in the figures.

Figure~\ref{fig:time} shows the computation time for each benchmark.
Due to the time-space tradeoff of a succinct data structure, our algorithm was a few times slower than the offline and SA.
The increase in computation time is acceptable for each case.

\begin{table}[t]
\begin{center}
\caption{Statistical information of benchmark string $S$}
\label{tab:txtinfo}
\begin{tabular}{ccccccc}
\hline
             & {\bf einstein}  & {\bf cere}  & {\bf kernel} & {\bf english} & {\bf dna}   & {\bf sources}     \\ \hline
$|S|$ (MB)       & $446$    & $446$ & $246$  & $200$   & $200$ & $200$       \\ 
$|\Sigma|$                & $139$    & $5$   & $160$  & $239$   & $16$ & $230$        \\ \hline
\end{tabular}
\end{center}
\end{table}

\begin{table}[t]
\begin{center}
\caption{Length of optimal $P$ extracted by suffix array (SA) 
and approximate $X$ by proposed algorithm (PA)
with approximation ratio $\frac{|val(X)|}{|P|}$ (\%) for top-100 patterns.} 
\label{tab:result}
\begin{tabular}{rrrrrrrr}
\hline
         &    & {\bf einstein}  & {\bf cere}  & {\bf kernel} & {\bf english} & {\bf dna}   & {\bf sources}     \\ \hline
         &SA& $198,606$ & $4,562$   & $442,124$   & $43,985$  & $3,271$  & $4,776$    \\ 
min.     &PA&  $18,625$ & $4,096$   & $37,205$    & $3,382$   & $268$    & $477$      \\ 
          &\% &  ${\bf 7.6}$ & ${\bf2.3}$   & ${\bf 6.9}$    & ${\bf 7.3}$   & ${\bf 7.1}$    & ${\bf 7.3}$      \\ \hline           
          &SA & $935,920$ & $303,204$ & $2,755,550$ & $98,7770$ & $97,979$ & $307,871$  \\ 
max.     &PA & $342,136$ & $58,906$  & $662,630$   & $16,1320$ & $24,834$ & $57,508$   \\
           &\% & ${\bf 50.0}$ & ${\bf 62.1}$  &$ {\bf 52.8}$   & ${\bf 50.8}$ & ${\bf 63.9}$ & ${\bf 51.7}$   \\ \hline 
           &SA & $259,451$ & $111,284$ & $727,443$   & $116,920$ & $8,241$  & $14,498$   \\ 
mean    &PA & $56,584$  & $12,723$  & $152,903$   & $24,703$  & $1,926$  & $3,279$    \\
           &\% & ${\bf 21.6}$  & ${\bf 11.0}$  & ${\bf 20.0}$   & ${\bf 23.0}$  & ${\bf 22.9}$  & ${\bf 22.0}$    \\ \hline
\end{tabular}
\end{center}
\end{table}

\begin{figure}[tb]
 \begin{minipage}[b]{0.5\linewidth}
  \centering
  \includegraphics[keepaspectratio, scale=0.3]
  {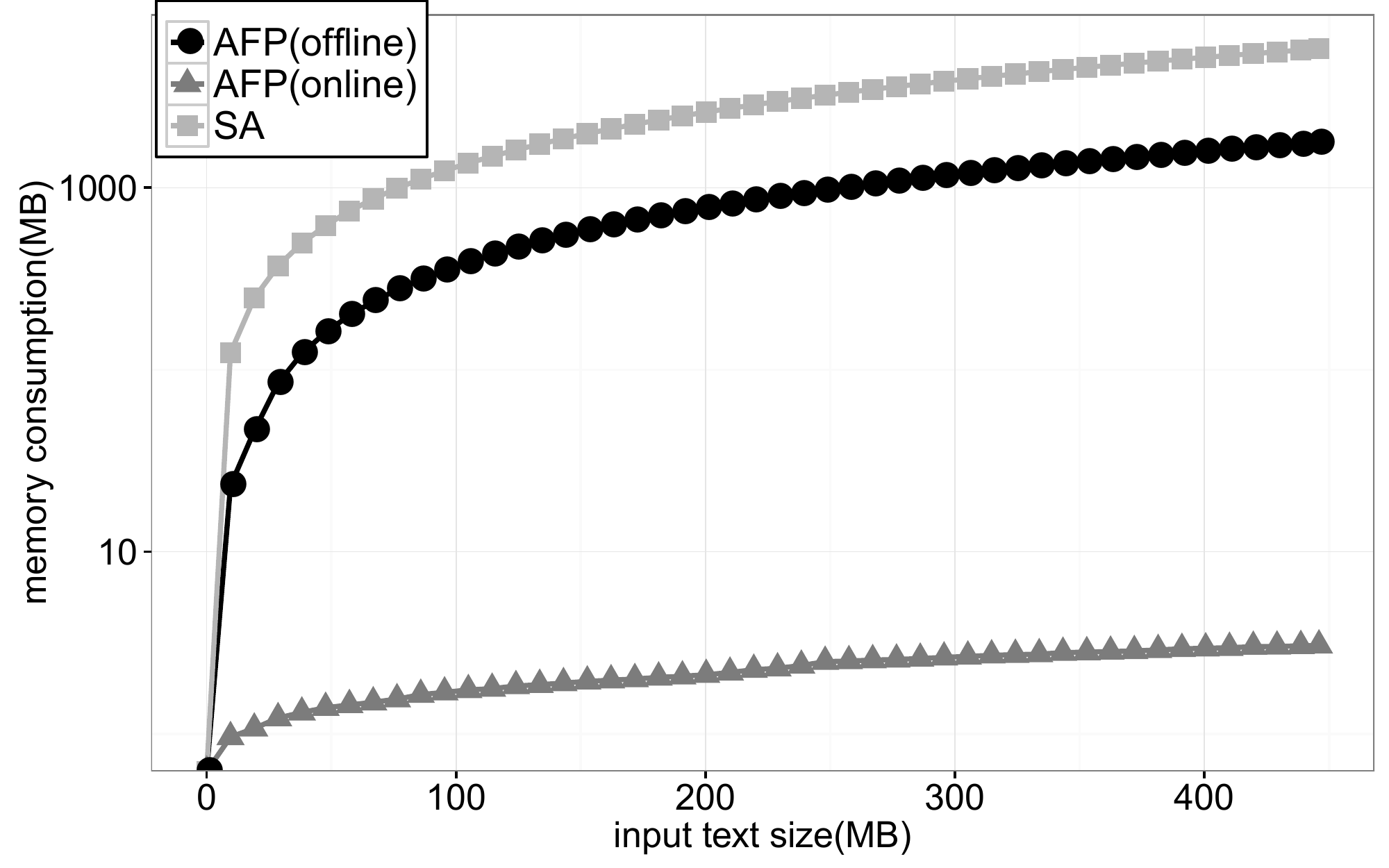}
  \subcaption{einstein}\label{fig:mem_ein}
 \end{minipage}
 \begin{minipage}[b]{0.5\linewidth}
  \centering
  \includegraphics[keepaspectratio, scale=0.3]
  {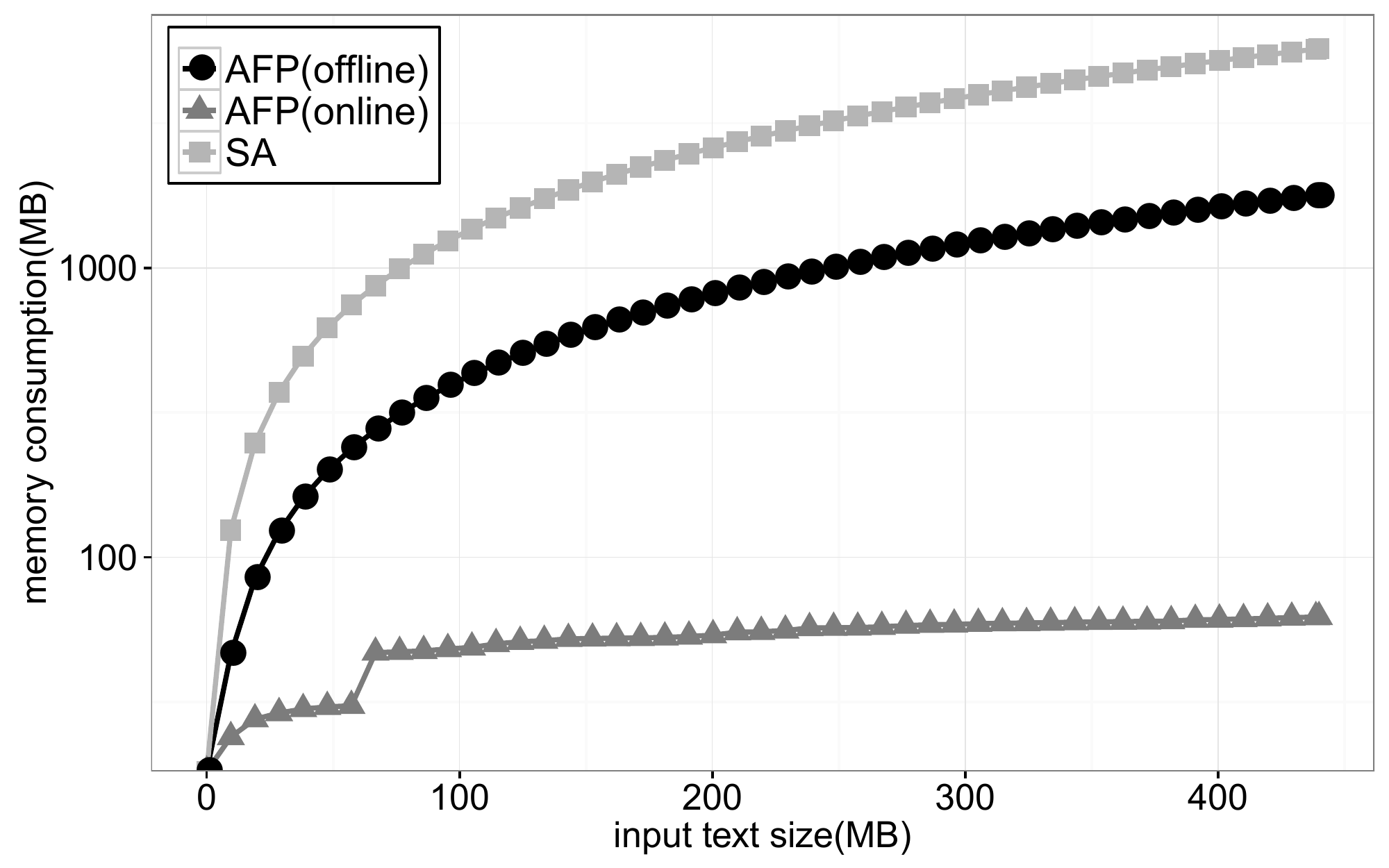}
  \subcaption{cere}\label{fig:mem_cere}
 \end{minipage}\\
 \begin{minipage}[b]{0.5\linewidth}
  \centering
  \includegraphics[keepaspectratio, scale=0.3]
  {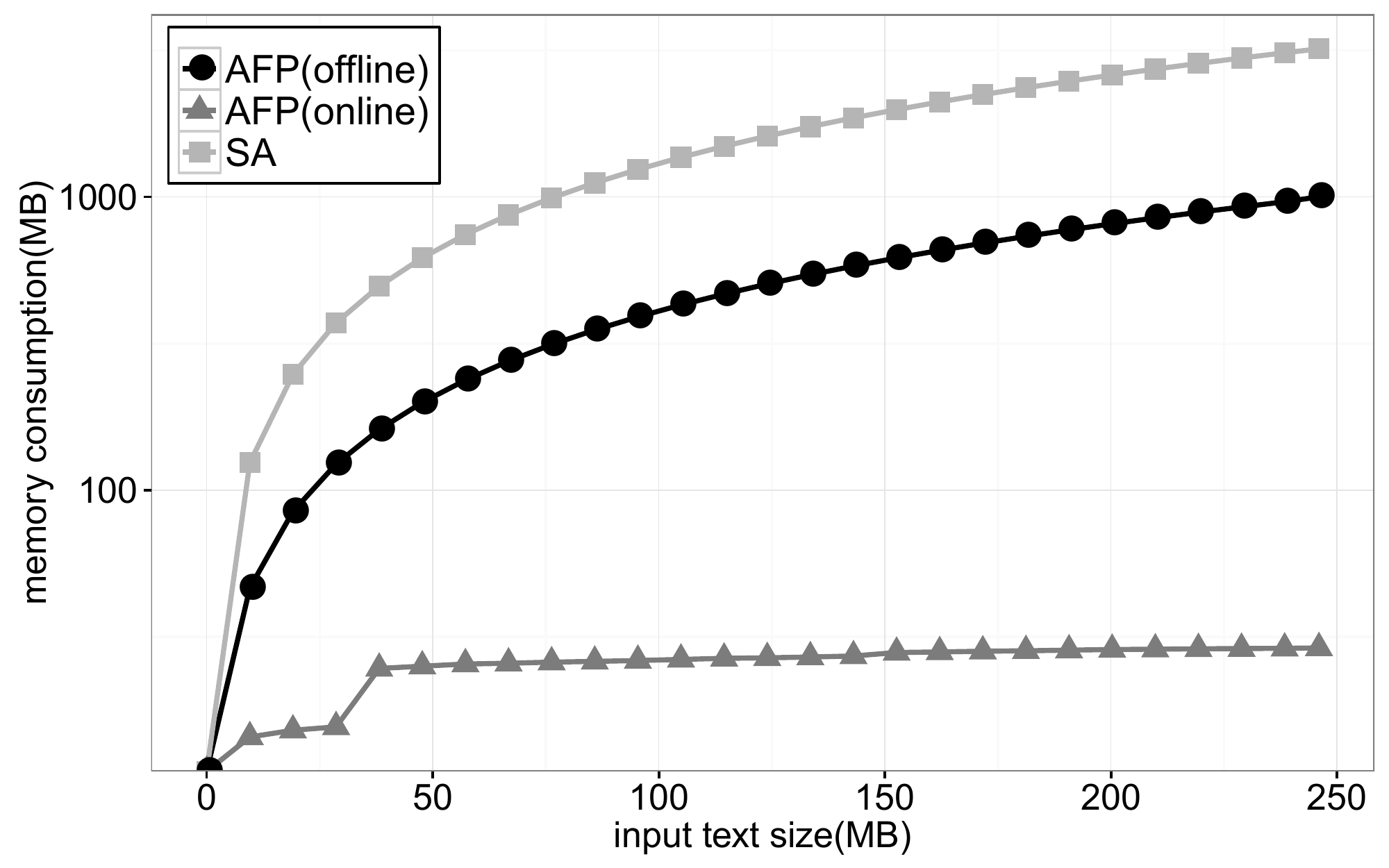}
  \subcaption{kernel}\label{fig:mem_ker}
 \end{minipage}
 \begin{minipage}[b]{0.5\linewidth}
  \centering
  \includegraphics[keepaspectratio, scale=0.3]
  {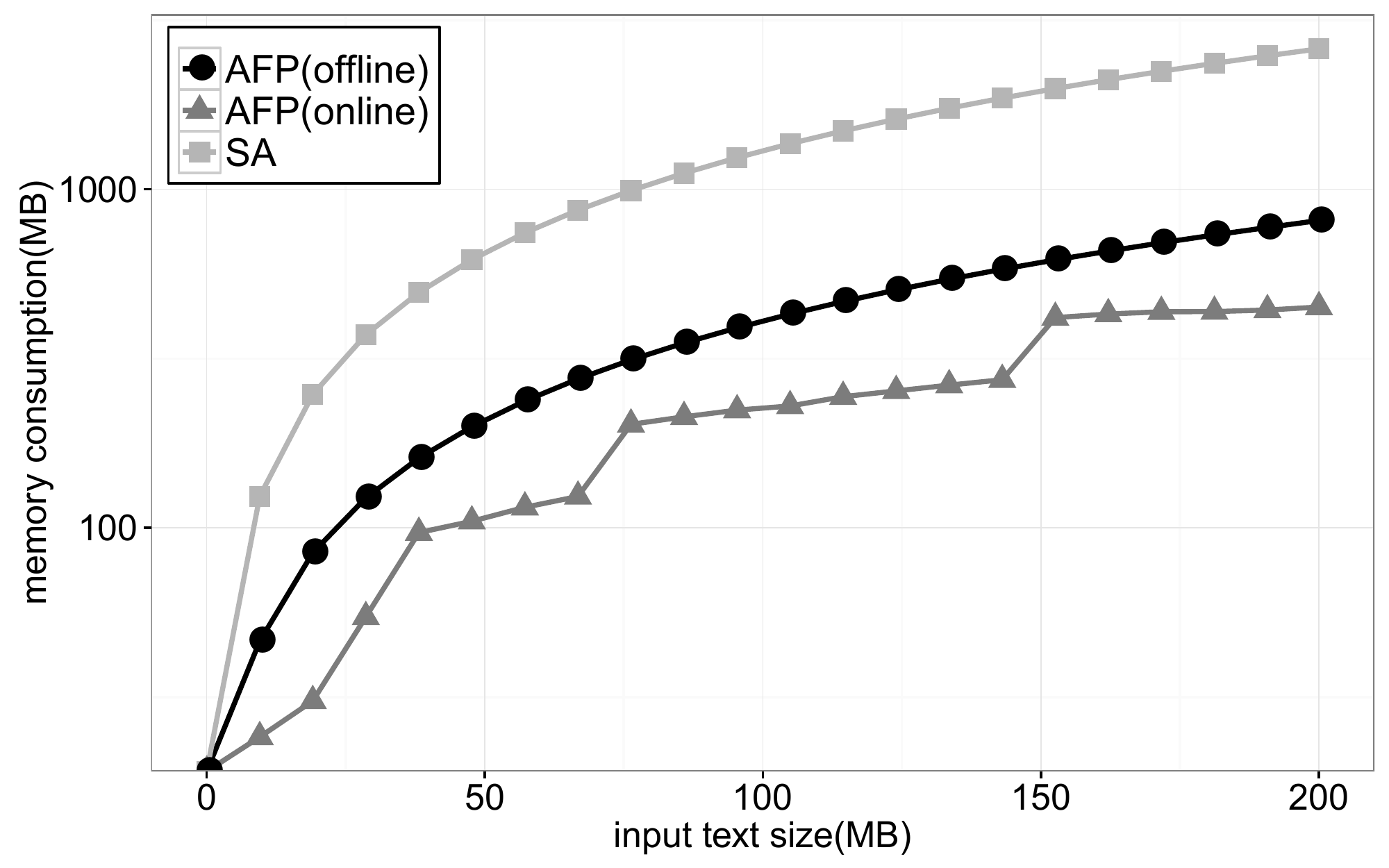}
  \subcaption{english}\label{fig:mem_eng}
 \end{minipage}\\
 \begin{minipage}[b]{0.5\linewidth}
  \centering
  \includegraphics[keepaspectratio, scale=0.3]
  {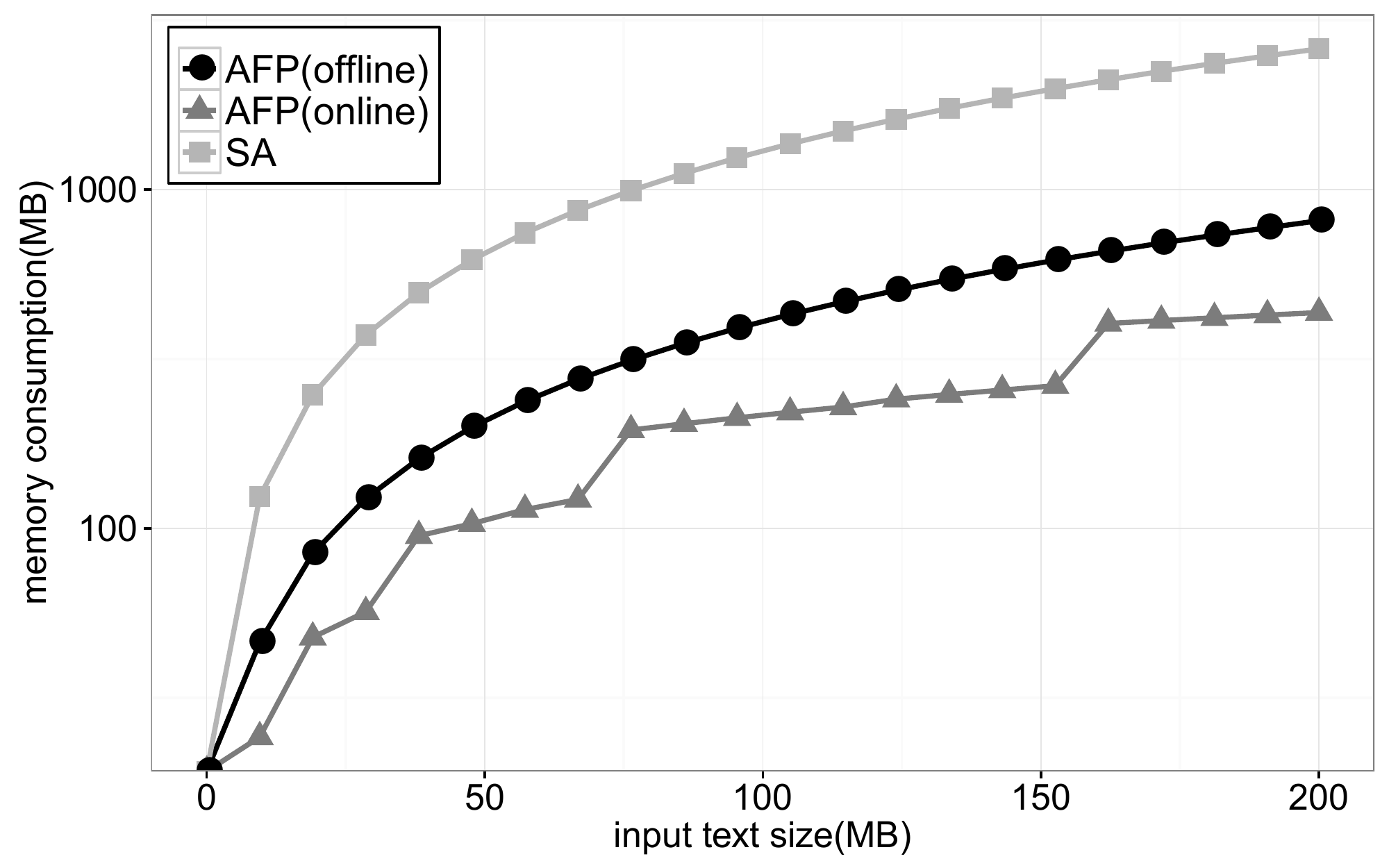}
  \subcaption{dna}\label{fig:mem_dna}
 \end{minipage}
 \begin{minipage}[b]{0.5\linewidth}
  \centering
  \includegraphics[keepaspectratio, scale=0.3]
  {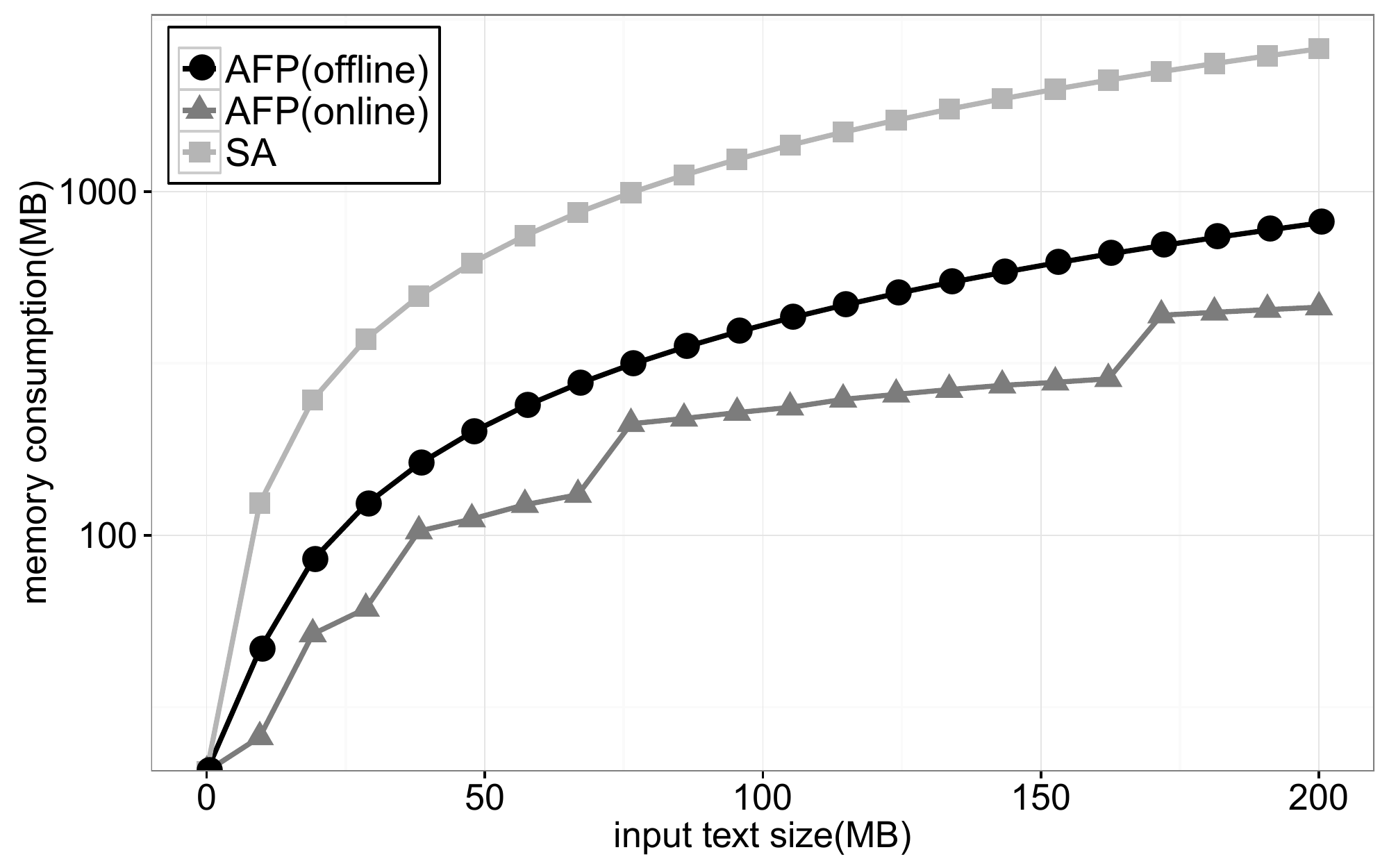}
  \subcaption{sources}\label{fig:mem_sou}
 \end{minipage}\\
\vspace{-.7cm}
 \caption{Memory consumption~(MB)}\label{fig:mem}
 \label{fig:mem}
\end{figure}

\begin{figure}[h]
 \begin{minipage}[b]{0.5\linewidth}
  \centering
  \includegraphics[keepaspectratio, scale=0.3]
  {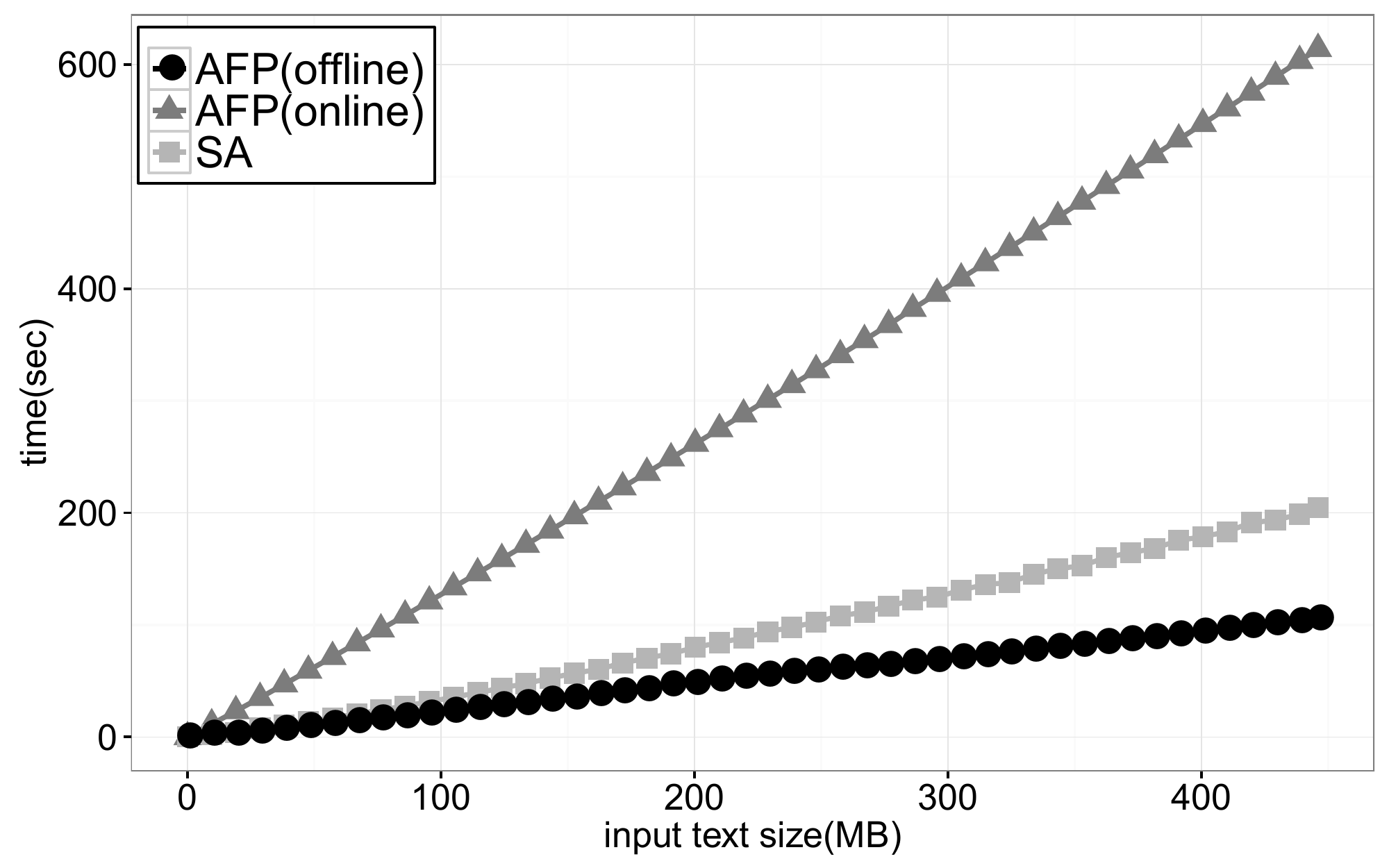}
  \subcaption{einstein}\label{fig:time_ein}
 \end{minipage}
 \begin{minipage}[b]{0.5\linewidth}
  \centering
  \includegraphics[keepaspectratio, scale=0.3]
  {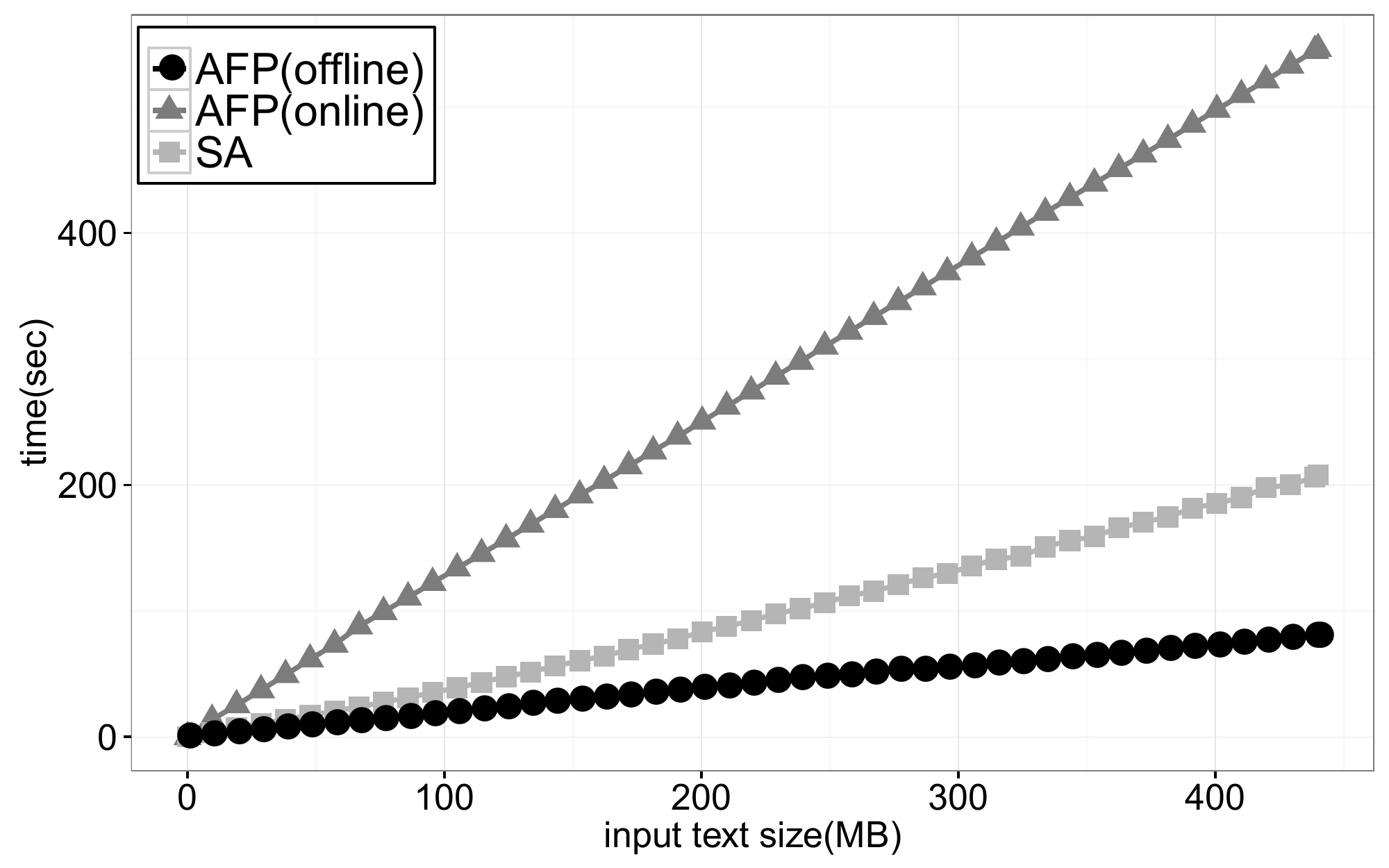}
  \subcaption{cere}\label{fig:time_cere}
 \end{minipage}\\
 \begin{minipage}[b]{0.5\linewidth}
  \centering
  \includegraphics[keepaspectratio, scale=0.3]
  {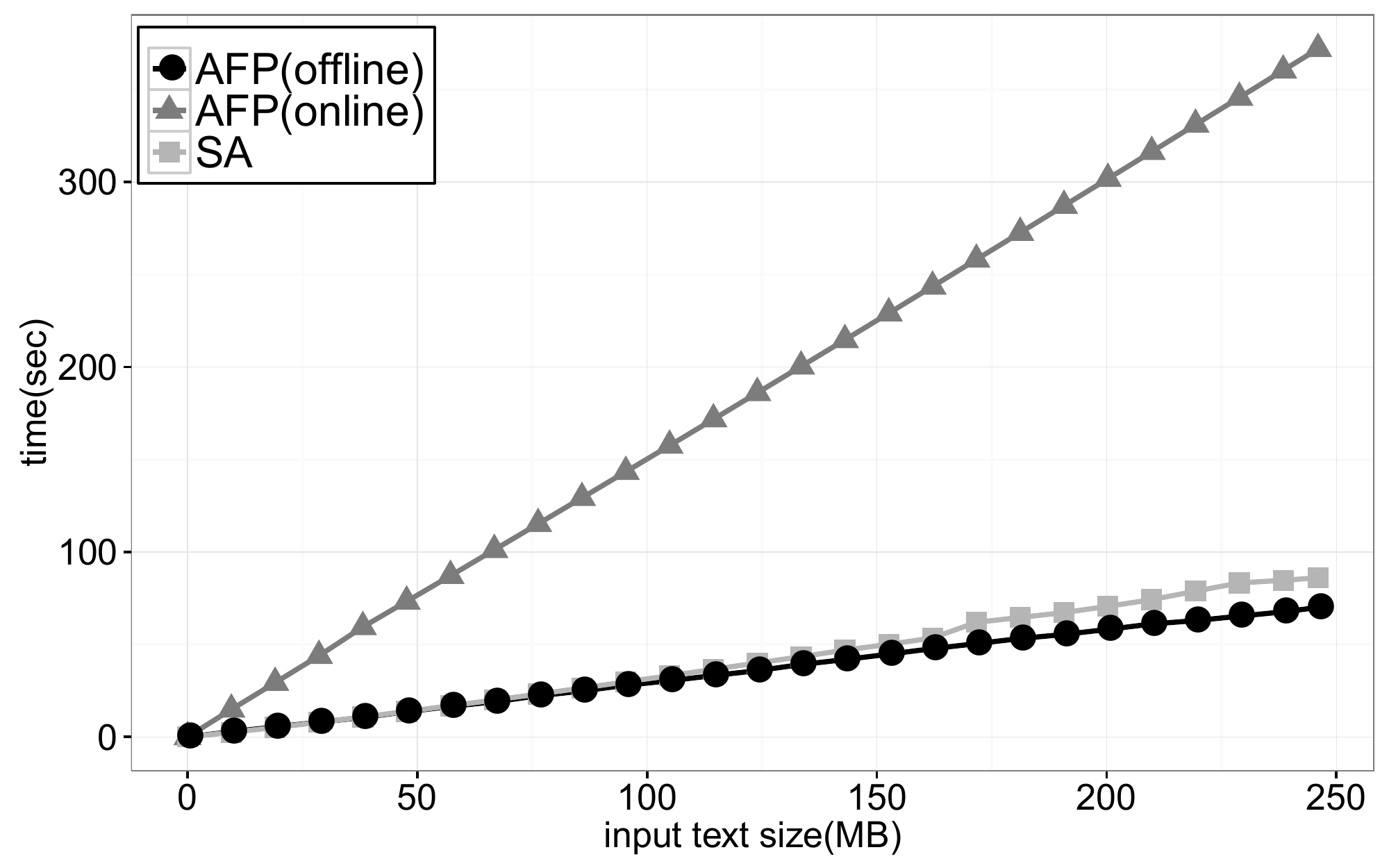}
  \subcaption{kernel}\label{fig:time_ker}
 \end{minipage}
 \begin{minipage}[b]{0.5\linewidth}
  \centering
  \includegraphics[keepaspectratio, scale=0.3]
  {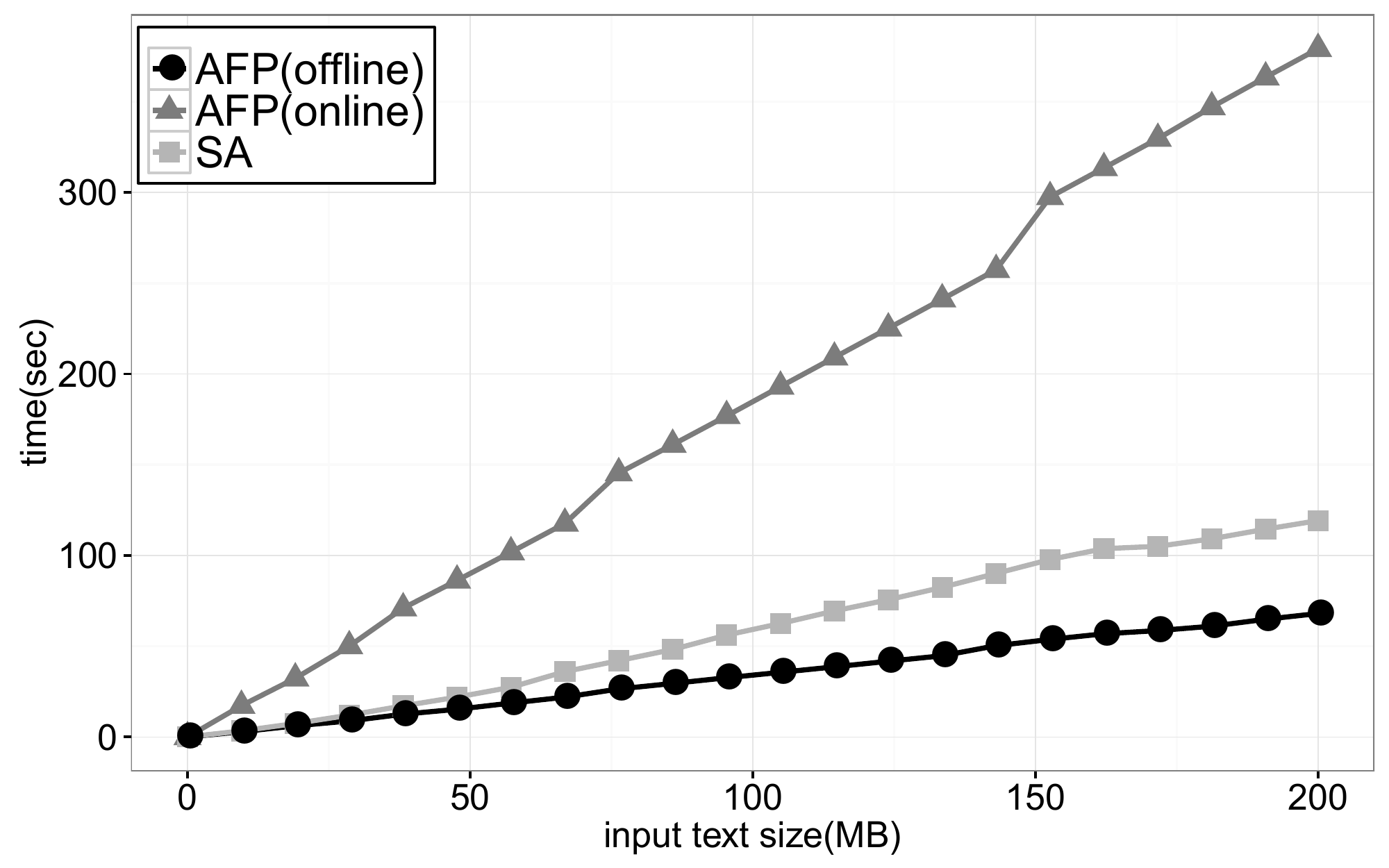}
  \subcaption{english.}\label{fig:time_eng}
 \end{minipage}\\
 \begin{minipage}[b]{0.5\linewidth}
  \centering
  \includegraphics[keepaspectratio, scale=0.3]
  {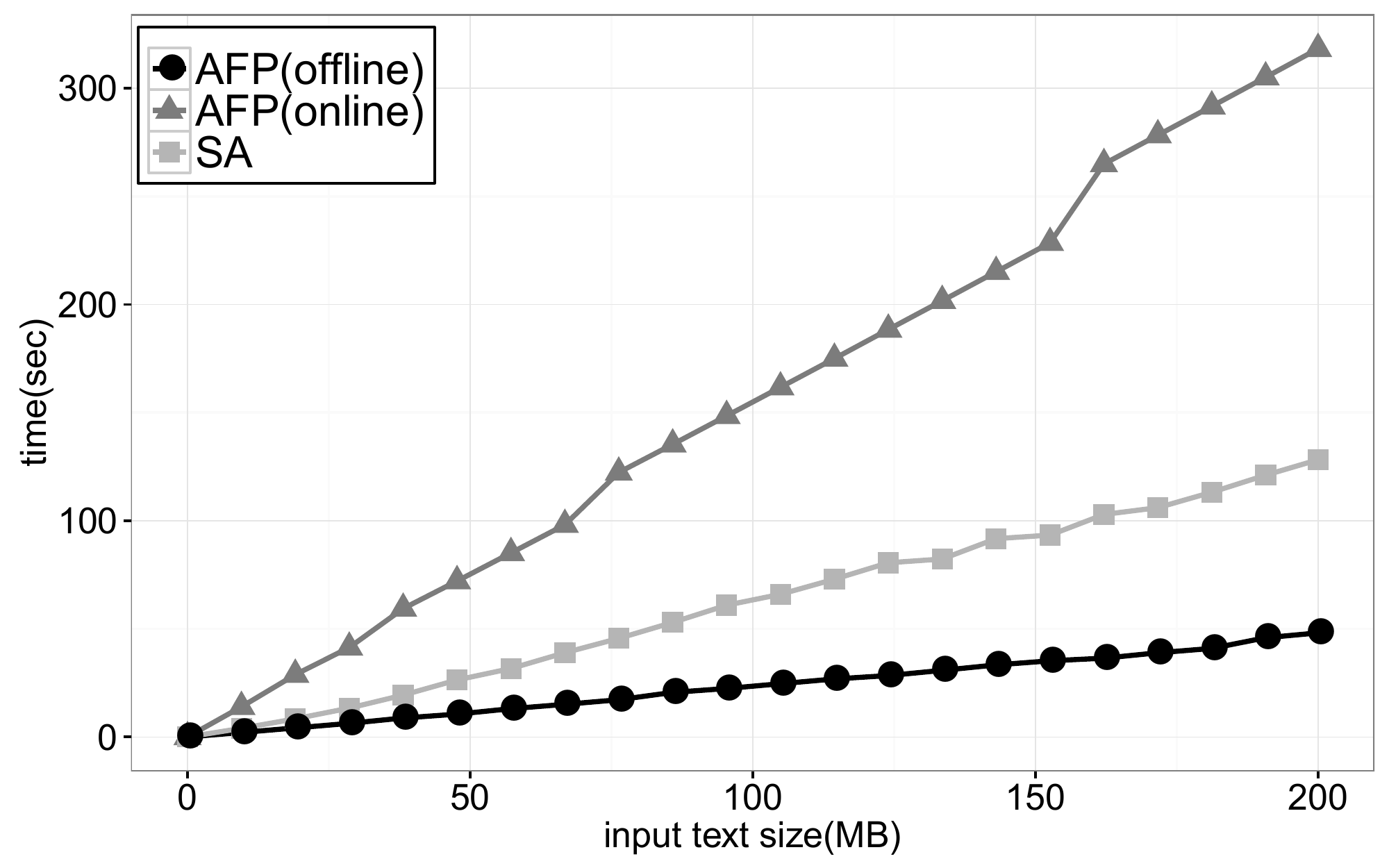}
  \subcaption{dna}\label{fig:time_dna}
 \end{minipage}
 \begin{minipage}[b]{0.5\linewidth}
  \centering
  \includegraphics[keepaspectratio, scale=0.3]
  {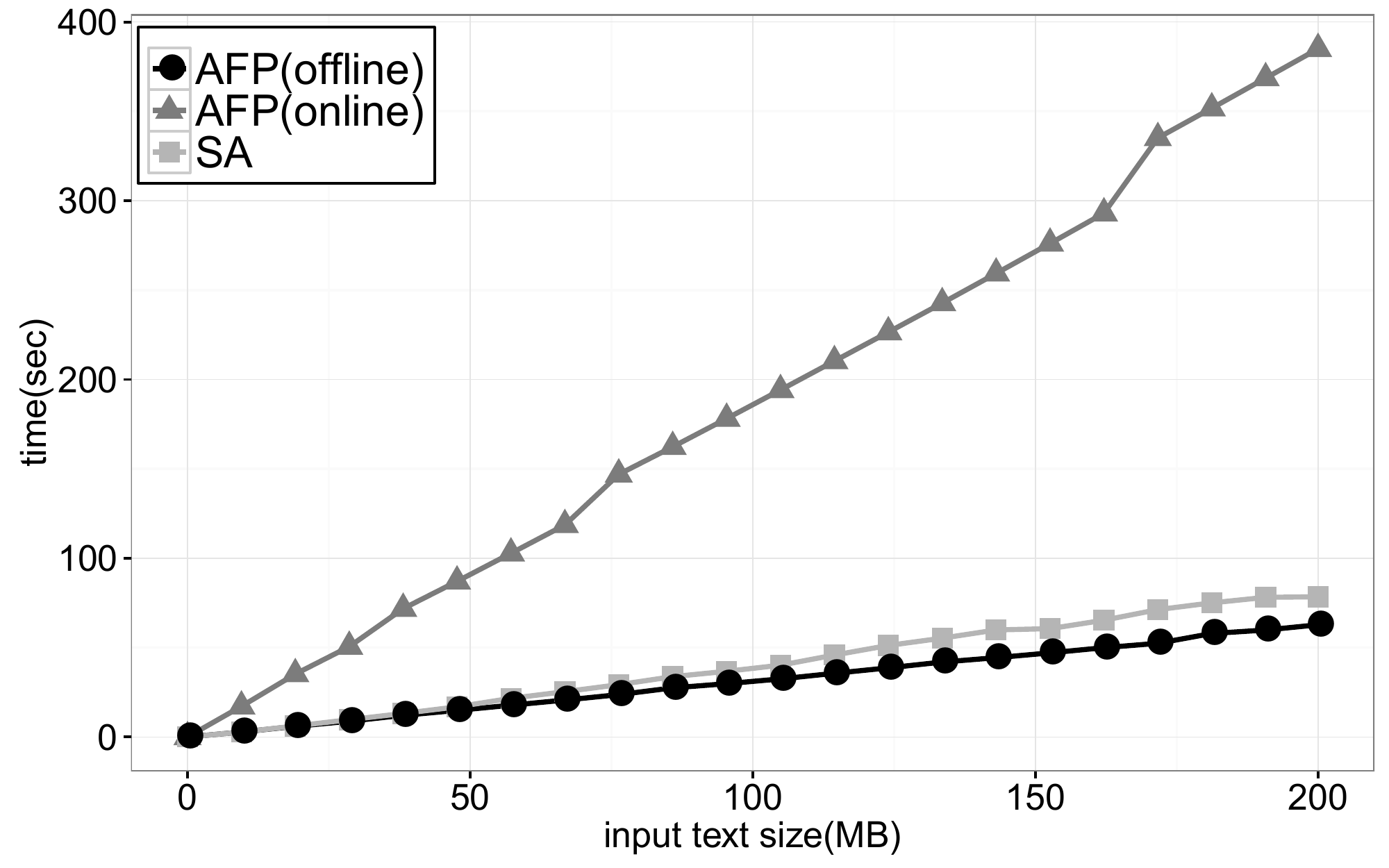}
  \subcaption{sources.}\label{fig:time_sou}
 \end{minipage}\\
\vspace{-.7cm}
 \caption{Computation time~(sec)}\label{fig:cmp}
 \label{fig:time}
\end{figure}

\section{Conclusion}
For the problem of finding frequent patterns, we proposed an online approximation algorithm with a compressed space.
Our algorithm is an improvement of FOLCA: a fully online grammar compression algorithm. 
We improved the theoretical lower bound of the approximation ratio
and presented experimental results exhibiting the efficiency for highly repetitive texts.
There is still a large gap of approximation ratio between theory and practical result.
An improvement of the lower bound is an important future work.

\bibliography{biblio,compress,book,cpm,succinct,math,other}
\bibliographystyle{plain}

\end{document}